\documentclass[twocolumn,aps,prb,superscriptaddress]{revtex4}

\usepackage{amsmath,amsthm,amsfonts,xspace,verbatim,graphicx}

\begin{document}
\newcommand{\op}[1]{#1}
\newcommand{\ket}[1]{\left| #1 \right\rangle}
\newcommand{\bra}[1]{\left\langle #1 \right|}
\newcommand{\braket}[2]{\left\langle #1 | #2 \right\rangle}
\newcommand{\braopket}[3]{\bra{#1}#2\ket{#3}}
\newcommand{\proj}[1]{| #1\rangle\!\langle #1 |}
\newcommand{\expect}[1]{\left\langle#1\right\rangle}
\newcommand{\Entropy}{H}
\newcommand{\KL}[2]{S\left(#1\|#2\right)}
\newcommand{\Tr}{\mathrm{Tr}}
\newcommand{\Rho}{P}
\def\Id{1\!\mathrm{l}}
\newcommand{\cM}{\mathcal{M}}
\newcommand{\cR}{\mathcal{R}}
\newcommand{\Rhat}{\hat{\mathcal{R}}}
\newcommand{\Vbar}{\overline{V}}
\newcommand{\cE}{\mathcal{E}}
\newcommand{\cL}{\mathcal{L}}
\newcommand{\cH}{\mathcal{H}}
\newcommand{\cU}{\mathcal{U}}
\newcommand{\cP}{\mathcal{P}}
\newcommand{\reals}{\mathbb{R}}
\newcommand{\grad}{\nabla}
\newcommand{\rhohat}{\hat{\op\rho}}
\newcommand{\rhoMLE}{\rhohat_\mathrm{MLE}}
\newcommand{\rhotomo}{\rhohat_\mathrm{tomo}}
\newcommand{\rhotrue}{\rho}
\newcommand{\pvec}[1]{{\bf #1}}
\newcommand{\pMLE}{\pvec{p}_{\mathrm{MLE}}}
\newcommand{\pHMLE}{\pvec{p}_{\mathrm{H}}}
\newcommand{\diff}{\mathrm{d}\!}
\newcommand{\pdiff}[2]{\frac{\partial #1}{\partial #2}}
\def\FCW{1.0\columnwidth}
\def\HCW{0.55\columnwidth}
\def\TPW{0.33\textwidth}

\newtheorem{lemma}{Lemma}
\newtheorem{theorem}{Theorem}
\newtheorem{definition}{Definition}

\graphicspath{{./}{Figures/}}

\title{Robust error bars for quantum tomography}

\author{Robin Blume-Kohout}
\affiliation{Theoretical Division (T-4/CNLS), MS-B258, Los Alamos National Laboratory, Los Alamos, NM 87545}
\email{robin@blumekohout.com}

\maketitle

\textbf{In quantum tomography \cite{Tomography}, a quantum state or process is estimated from the results of measurements on many identically prepared systems.  Tomography can never identify the state ($\rho$) or process ($\cE$) exactly.  Any point estimate is necessarily ``wrong'' -- at best, it will be close to the true state.  Making rigorous, reliable statements about the system requires \emph{region estimates}.  In this article, I present a procedure for assigning \emph{likelihood ratio (LR) confidence regions}, an elegant and powerful generalization of error bars.  In particular, LR regions are almost optimally powerful -- i.e., they are as small as possible.}

Quantum information processing relies on quantum hardware, including memory qubits and unitary or nearly-unitary quantum gates.  These individual components must perform their allotted transformations with very high precision, especially for fault-tolerant quantum computing.  The methods used to characterize and validate quantum devices are known, collectively, as \emph{quantum tomography}.  Tomography usually involves repeated independent measurements on $N$ identically-prepared systems (referred to hereafter as ``standard tomography''), but can also involve collective measurements on all $N$ copies.  Because state and process tomography are mathematically equivalent, this paper will focus on state tomography for the sake of clarity, with the understanding that all results can be extended straightforwardly to processes\footnote{This is done by identifying a quantum process $\cE$ with a bipartite state $\rho(\cE) = \left(\Id\otimes\cE\right)[\proj{\Psi}]$, where $\ket{\Psi}$ is a maximally entangled state.  The differences between state and process tomography are entirely in implementation, and they all pertain to \emph{gathering} data, not to its analysis.}.

Tomography cannot identify $\rho$ (the state produced by a quantum device) \emph{exactly}, for precisely the same reason that flipping a coin $N$ times cannot reveal its bias exactly.  Any point estimate $\rhohat$ has precisely zero probability of coinciding exactly with the true $\rho$, for there are infinitely many other states arbitrarily close to $\rhohat$ and equally consistent with the data.  To make a tomographic assertion about the device that \emph{is} true -- or at least true with high probability -- we must report a \emph{region} of states or processes (vis. Fig. \ref{fig:RegionEstimate}).

\begin{figure}[t]
\begin{centering}
\includegraphics[width=8cm]{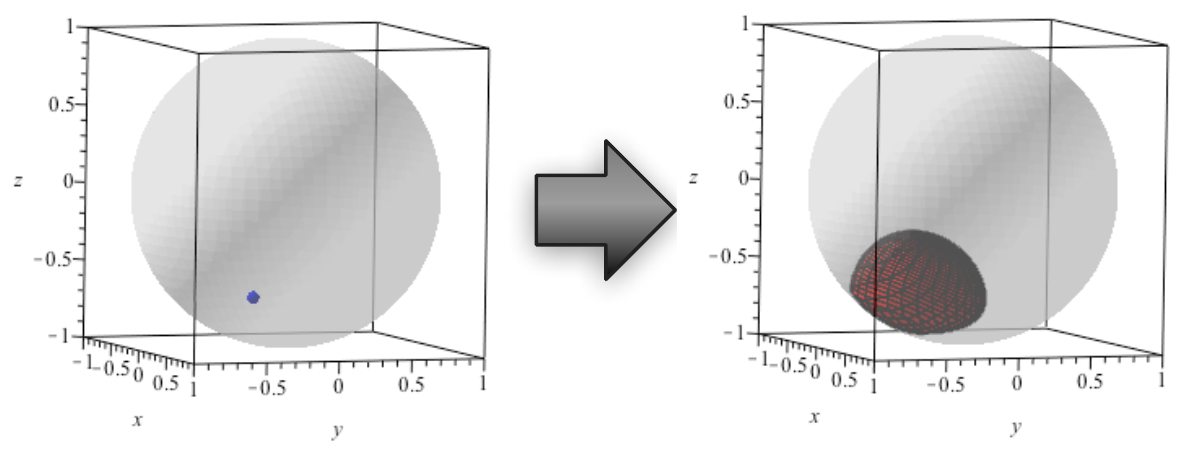}
\par\end{centering}
\caption{\emph{Point estimators}, like the maximum likelihood estimate $\rhoMLE$ shown on the left, cannot provide meaningful and rigorous statements about the true (but unknown) state $\rho$.  But if we replace point estimators with \emph{region estimators}, like the likelihood-ratio confidence region shown on the right, then the region $\Rhat$ defines an assertion -- ``$\rho$ lies within $\Rhat$ with 90\% certainty'' -- that is rigorously valid.  The estimates shown here came from simulated measurements on 60 copies of a single-qubit state, with 20 measurements each of $\sigma_x,\sigma_y,\sigma_z$ yielding +/- counts of 7/13, 9/11, and 3/17 (respectively).}
\label{fig:RegionEstimate}
\end{figure}

Such regions are often constructed by attaching \emph{error bars} to a point estimate.  In quantum tomography, this approach suffers several drawbacks, some of which are illustrated in Fig. \ref{fig:CircleRegion}.  Na\"ive error bars define an ellipsoidal shape (arbitrary), centered at the point estimate (suboptimal), which may include many unphysical states (inefficient).  Worst of all, it is generally impossible to assign this ellipsoid any rigorous meaning -- e.g., ``The true state is within it, with probability at least 99\%.''  The same problem applies to the other method used to date, \emph{bootstrapping}\cite{HomeScience09} -- which means generating a host of simulated datasets $\{D_k\}$ (either by resampling the real data, or by simulating measurements on a point estimate $\rhohat$), then reporting the variance of the corresponding point estimates $\{\rhohat_k\}$.  The underlying problem is that bootstrapping and na\"ive error bars both represent \emph{standard errors} -- the variance of a point estimator.  Unfortunately, the point estimators used in quantum tomography are all \emph{biased}, and standard errors for biased estimators do not reliably represent uncertainty\footnote{To see this, consider a simple biased point estimator that assigns $\rhohat=\Id/d$ no matter what data is observed.  Clearly, the variance of $\rhohat(D)$ is zero -- and, just as clearly, this implies nothing about our uncertainty about the true $\rho$!  While obviously extreme and even absurd, this estimator nonetheless demonstrates that the variance of biased estimators cannot be relied upon to describe uncertainty.  In real-world tomography, the maximum likelihood estimator $\rhoMLE$ is biased by the positivity constraint $\rho\geq0$, and this \emph{can} lead to significant underestimation of uncertainty.} about the true $\rho$.

\begin{figure}[b]
\begin{centering}
\includegraphics[width=8cm]{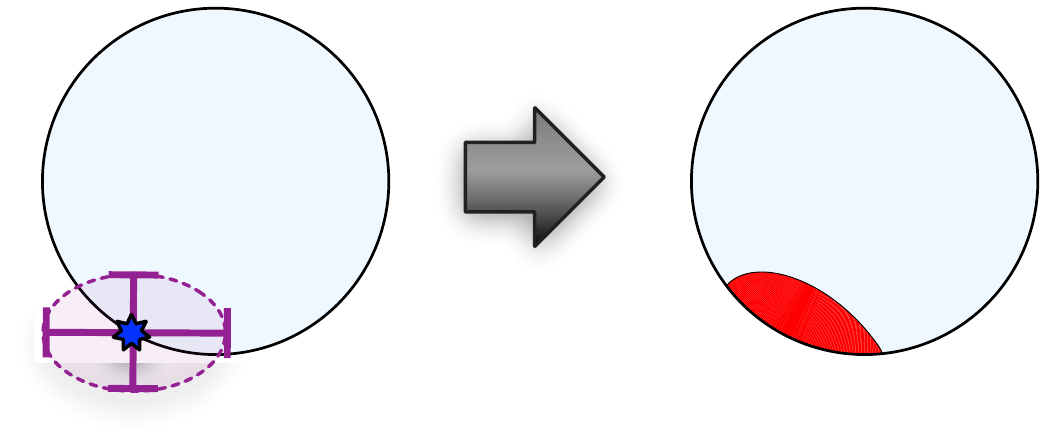}
\par\end{centering}
\caption{General region estimates -- adapted to the data, and constructed so as to minimize volume -- can be far more powerful, useful, and reliable than traditional ``error bars''.  As illustrated here, a valid confidence region need not be: (i) ellipsoidal or rectangular, (ii) centered at a point estimate, or (iii) aligned with the axes defined by whatever observables were measured.  The figure on the right shows a cross-section of a 1-qubit LR confidence region, while on the left the smallest traditional error bars with the same coverage probability are shown.  The LR region is noticeably smaller, and includes only valid states.  Although in this case, the LR region could be reasonably approximated by the intersection between the error ellipsoid and the Bloch sphere, this is not always the case.}
\label{fig:CircleRegion}
\end{figure}

Happily, all of these issues can be resolved with a remarkably simple construction.  \emph{Likelihood ratio (LR) confidence regions} (see Fig. \ref{fig:Examples} for some examples) generalize the notion of error bars, providing data-adapted (not necessarily ellipsoidal) regions that:
\begin{enumerate}
\item contain the true state with guaranteed, high, and user-specified probability;
\item are, on average, smaller (and thus more powerful) than almost any other construction; and
\item are simple to define and construct.
\end{enumerate}

\begin{definition} Given observed data $D$, the \textbf{likelihood} is a function on states given by $\cL(\rho) = Pr(D|\rho)$.  The \textbf{log likelihood ratio} is a function on states given by $\lambda(\rho) = -2\log\left[ \cL(\rho) / \max_{\rho'}\cL(\rho')\right]$.  Given data $D$, the \textbf{likelihood ratio region} with confidence $\alpha$ is $\Rhat_\alpha(D) = \{\mathrm{all\ }\rho\mathrm{\ such\ that\ }\lambda(\rho) < \lambda_\alpha\}$, where $\lambda_\alpha$ is a constant (see below) that depends on the desired confidence $\alpha$ and the Hilbert space dimension $d$. \label{def:LR}
\end{definition}

It should be obvious that the threshold value $\lambda_\alpha$ plays a critical role in this construction.  Increasing $\lambda_\alpha$ increases the size of the LR region $\Rhat_\alpha$, which in turn increases the estimator's \emph{coverage probability} -- the probability that $\Rhat(D)$ will contain $\rho$ -- but reduces its \emph{power} (since large regions imply less about $\rho$).  So $\lambda_\alpha$ should be set to the smallest value that ensures coverage probability at least $\alpha$.  This optimal value depends on $\alpha$ and the size of the system we are measuring.  It is hard to compute exactly, but two upper bounds are provided in Eq. \ref{eq:Fbound} and Lemma \ref{lem:bound}.  Either will guarantee coverage probability at least $\alpha$, at the cost of slightly increasing the regions' size.

\begin{figure}[t]
\begin{centering}
\includegraphics[width=8cm]{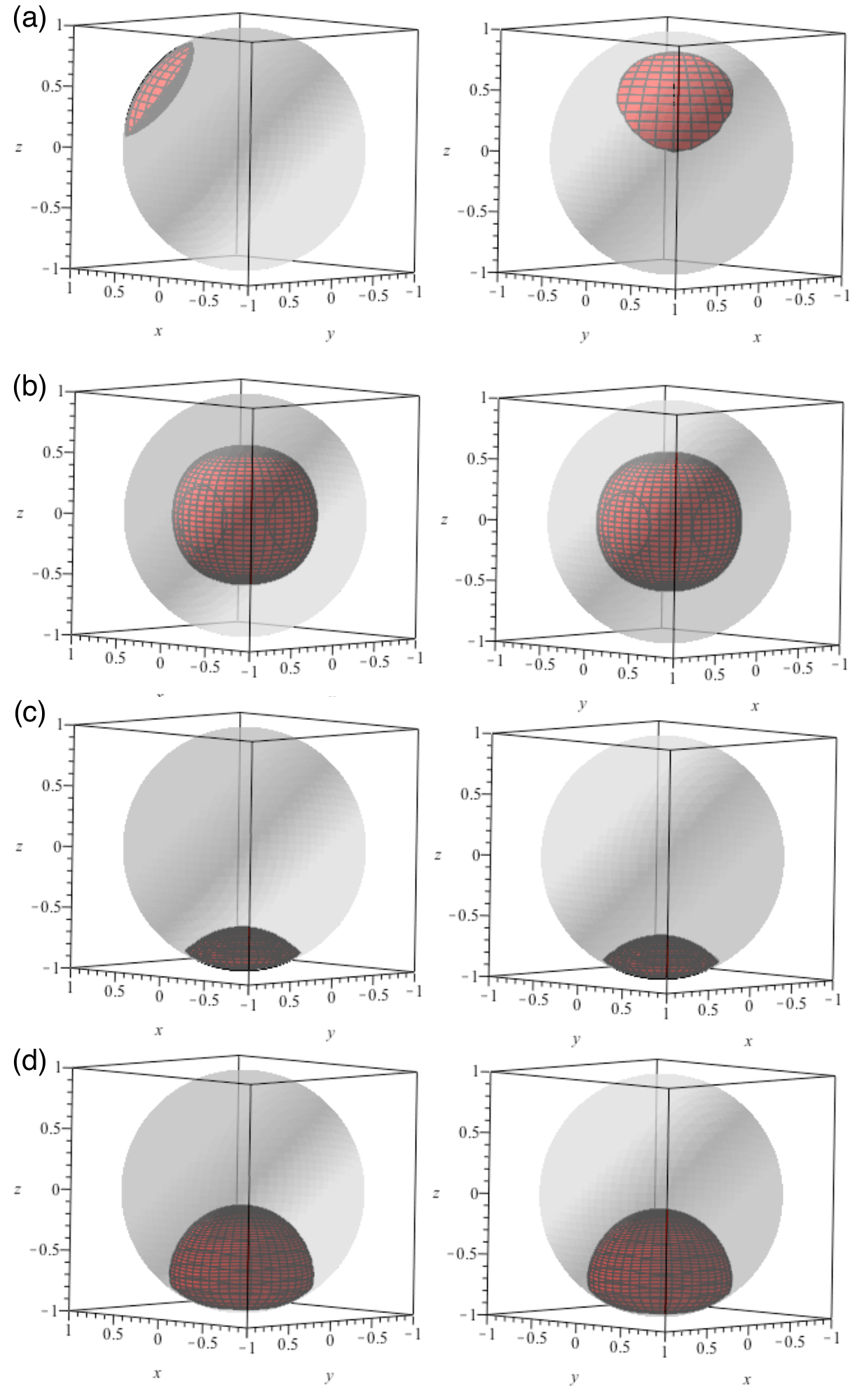}
\par\end{centering}
\caption{Four examples of the shapes that LR regions can take.  Each example is a 90\% confidence region for a qubit, based on 60 measurements divided equally among the three Pauli observables.  Rows (a-d) correspond to four distinct datasets.  Left and right columns are different views of the same region.}
\label{fig:Examples}
\end{figure}

The remainder of this paper attempts to answer three natural questions, in order of increasing technicality.  First, ``What is a confidence region, and how does it generalize `error bars'?''  Next, ``Why is the LR construction an especially good one?''  Finally, ``How do we choose the threshold?''  The concluding discussion section addresses a few other questions, especially the relationship to related work by Christandl and Renner\cite{Christandl}, and how to use and describe LR regions.

\section{Region estimators and confidence regions}

Error bars around a point estimate define a region estimate, but region estimates do not need to be associated with a point estimate.  After seeing the data ($D$), we can assign a region $\Rhat(D)$ of whatever shape and size is necessary to achieve our goals.  So what are these goals?

First and foremost is \emph{coverage probability}.  By assigning a region, we assert that the unknown $\rho$ is within it.  This had better be true with very high probability.  Coverage probability ($\alpha$) is the probability that $\Rhat(D)$ does indeed contain $\rho$.  It would be very satisfying indeed if we could assert ``Given the observed data, the probability that $\rho$ is in $\Rhat(D)$ is $\alpha$,'' i.e.
$$
Pr(\rho\in\Rhat(D)|D) \geq \alpha.
$$
Unfortunately, this assertion\footnote{This condition defines a Bayesian \emph{credible region}.} \emph{requires} assigning a \emph{prior} probability to ``$\rho\in\Rhat(D)$''.  Different agents (e.g., a scientist, a skeptical reader, and a funding agency) will generally disagree about this prior, and therefore about $Pr(\rho\in\Rhat(D)|D)$.  

So, instead, we make a subtly different assertion:  ``The region $\Rhat(D)$ that we assign will contain $\rho$ with probability at least $\alpha$,'' i.e. 
\begin{equation}
Pr(\rho\in\Rhat(D))\geq\alpha.
\end{equation}
This assertion is made \emph{before} the data $D$ are observed.  Once we know $D$, $\Rhat(D)$ is fixed, and the most that we can say is ``This region was obtained by a procedure that `works' almost always (i.e., with probability $\geq\alpha$)''.  Now (after the estimate), $\alpha$ is not the probability of success -- for success is no longer a random variable.  Instead, $\alpha$ quantifies our \emph{confidence} in the estimate, and so an estimation procedure satisfying this condition is a \emph{confidence region} estimator.

Confidence is a property of the entire estimator -- the map from data to regions -- rather than of a particular region estimate.  It depends on the unknown state -- but, \emph{mirabile dictu}, we can place relatively tight prior-\emph{independent} bounds on it,
$$Pr(\rho\in\Rhat(D))\  \geq\  \min_\rho{ Pr(\rho\in\Rhat(D)|\rho) }.$$
A confidence region estimator with confidence $\alpha$ satisfies
$$\Pr(\rho\in\Rhat(D)|\rho) \geq \alpha\ \forall\ \rho.$$
It's important to understand that confidence regions do not provide probabilistic statements about any single run of the experiment.  Once the data are taken, the estimator either succeeded or failed, and there is no way to assign a probability to its success without choosing a prior.  In any given experiment, what can be said is ``We applied a technique which is guaranteed to yield a region containing the true $\rho$ at least $\alpha$ of the time -- no matter what the unknown $\rho$ is.''

\section{Optimality}

Confidence regions are a basic statistical construct, especially for scalar parameters (where they are known as confidence \emph{intervals}).  Yet even for 1-dimensional intervals, there exist many distinct constructions, and no consensus on the ``best'' choice.  Note that designing a high-confidence region estimator isn't hard.  For example, the estimator $\Rhat(D) = \{\mathrm{all\ states}\}$ has coverage probability $\alpha=1$.  The challenge is to design one that is \emph{powerful} -- i.e., assigns small regions (which correspond to powerful hypotheses, because they rule out many states).

The likelihood ratio construction in Definition \ref{def:LR} was introduced in 1989 in the context of particle physics by Feldman and Cousins \cite{FeldmanCousins}, and appears to have been part of statistical folklore before that.  What has \emph{not} appeared to date is a compelling argument why LR regions are particularly good.  In this section, I provide such a justification, comprising: (1) a proof that the most powerful confidence region estimators are \emph{probability ratio} (PR) estimators (a similar result was proven by Evans et al\cite{Evans}; the treatment here is self-contained), and (2) a heuristic argument that, among the family of PR estimators, LR estimators have nearly-optimal worst-case behavior.

First, we need to quantify the power of any given region estimator $\Rhat(\cdot)$.  Smaller regions are clearly more powerful, and in this work I will quantify a region's power by its \emph{volume},
\begin{equation}
V(\Rhat) = \int_{\rho\in\Rhat}{\diff\rho}.\label{eq:defVol}
\end{equation}
Choosing a particular volume measure $\diff\rho$ could be controversial.  Remarkably, the construction in Definition \ref{def:PR} is optimal for \emph{any} measure $\diff\rho$!

Now, the volume of the assigned region is itself a random variable, depending on the data.  Averaging over datasets yields an expected volume,
$$\Vbar(\rho) = \sum_D{Pr(D|\rho)V(\Rhat(D))},$$
which is a function of $\rho$.  Since $\rho$ is (by definition) unknown, we can quantify the estimator's performance either by \emph{worst-case} (maximum) or \emph{average} volume.  To average, we must choose a measure $\mu = P(\rho)\diff\rho$ over $\rho$ (which need \emph{not} be related in any way to $\diff\rho$).  Optimal average performance is achieved by that measure's \emph{Bayes estimator}\cite{Robert2007}.

\begin{definition}
Given a cost function $V$, the \textbf{Bayes estimator} for a given measure $\mu$ (over states) is the estimator with the smallest average expected cost w/r.t. $\mu$.
\end{definition}

But the Bayes estimator for one measure might have very bad performance for another measure.  So if we are not sure what measure to choose, we have an alternative:  choose an estimator that minimizes worst-case performance, $\max_\rho{\Vbar(\rho)}$.  This defines the \emph{minimax} estimator\cite{Robert2007}.  Happily (and perhaps surprisingly), these two approaches are intimately related by a basic theorem of decision theory:

\begin{theorem} \textbf{(Minimax-Bayes duality\cite{Robert2007,Kempthorne87})} The minimax estimator (for a given cost function) is, under mild regularity conditions, also the Bayes estimator for some measure $\mu$, known as the \emph{least favorable prior}.
\end{theorem}

So to find the minimax estimator -- which is appealing because it has the best possible \emph{guaranteed} performance -- we will focus first on optimizing average performance
$$\expect{\Vbar}_P = \int{\Vbar(\rho)P(\rho)\diff\rho},$$
for an arbitrary measure $P(\rho)\diff\rho$.  (N.B.  In spite of this notation, the averaging measure $P(\rho)\diff\rho$ and volume measure $\diff\rho$ are \emph{completely} independent!  For example, Eq. \ref{eq:defPrD} doesn't depend on the choice of volume measure $\diff\rho$.)

A region estimator $\Rhat(\cdot)$ can be represented by a connection relation: for each dataset $D$ and state $\rho$, we say that $\rho$ is ``connected'' to $D$ ($\rho\sim D$) iff $\rho\in\Rhat(D)$.  The average expected volume is then given by
\begin{eqnarray}
\expect{\Vbar} &=& \int{\Vbar(\rho)P(\rho)\diff\rho} \nonumber \\
&=& \int{ \sum_D{Pr(D|\rho)V(\Rhat(D))}P(\rho)\diff\rho } \nonumber \\
&=& \sum_D{ Pr(D) V(\Rhat(D)) },
\end{eqnarray}
where 
\begin{equation}
Pr(D) \equiv \int{Pr(D|\rho)P(\rho)\diff\rho}\label{eq:defPrD}.
\end{equation}
Now, we recall Eq. \ref{eq:defVol} for the volume of a region, and observe that the integral over $\rho\in\Rhat(D)$ is equivalent to an integral over $\rho\sim D$, which gives
\begin{eqnarray}
\expect{\Vbar} &=& \sum_D{ Pr(D) \int_{\rho\sim D}{\diff\rho}} \nonumber \\
&=& \int_{\mathrm{all\ }\rho}{\left(\sum_{D\sim\rho}{ Pr(D) }\right)\diff\rho}. \label{eq:exVbar}
\end{eqnarray}
Now, recall that we want to minimize $\expect{\Vbar}$ (Eq. \ref{eq:exVbar}) subject to the constraint that
\begin{equation}\sum_{D\sim\rho}{Pr(D|\rho)} \geq \alpha\mathrm{\ for\ each\ }\rho. \label{eq:constraint}
\end{equation}
We can minimize each of the terms in the integral (Eq. \ref{eq:exVbar}) independently -- because they are not coupled by the constraint.  To do so, consider a simple cost/benefit analysis.  Eq. \ref{eq:constraint} says that each state $\rho$ must be connected to datasets whose total \emph{conditional} probability $Pr(D|\rho)$ is at least $\alpha$.  But Eq. \ref{eq:exVbar} says that each such connection comes at a cost, given by the \emph{unconditional} probability $Pr(D)$.  To achieve total conditional probability $\alpha$ at minimum cost, we connect $\rho$ to datasets in descending order of benefit/cost ratio, given by the \emph{probability ratio} statistic
\begin{equation}
r(D;\rho) = \frac{Pr(D|\rho)}{Pr(D)},\label{eq:PR}
\end{equation}
down to a threshold $r_\alpha(\rho)$ that satisfies
\begin{equation}
\sum_{\mathrm{all\ }D\mathrm{\ s.t.\ }r(D;\rho) > r_\alpha(\rho)}{Pr(D|\rho)} \geq \alpha.\label{eq:ralpha}
\end{equation}
Inverting this relationship (to define which states are connected to a given $D$, rather than the other way around) yields probability ratio (PR) region estimators:
	\begin{definition} The PR region estimator for an averaging measure $\mu = P(\rho)\diff\rho$, with confidence level $\alpha$, is defined by $\Rhat(D) = \{\mathrm{all\ }\rho{\ such\ that\ }Pr(D|\rho)/Pr(D) \geq r_\alpha(\rho)\},$ with $Pr(D)$ and $r_\alpha(\rho)$ given by Eqs. \ref{eq:defPrD},\ref{eq:ralpha}. \label{def:PR}
\end{definition}
This prescription is an exact solution to the problem of minimum-average-volume confidence regions -- and, as advertised, it does not depend on what measure $\diff\rho$ is used to defined volume!

PR estimators are interesting in themselves.  In particular, the estimator that unconditionally ``works best'' for $\rho=\rho_0$ (for any given $\rho_0$) is the PR estimator for the measure $\mu = \delta(\rho-\rho_0)$.  An especially confident experimentalist who believes that her apparatus really is producing $\rho_0$ might choose this estimator.  Despite the radical ``prior'', this estimator still assigns valid confidence regions, on which a skeptical third party can rely.  The experimentalist's extreme confidence is reflected only in this manner:  the datasets $D$ that typically occur when $\rho=\rho_0$ yield relatively small regions, while datasets $D$ that are improbable given $\rho=\rho_0$ (but might appear with high probability for other states) yield enormous regions.  Thus, if $\rho$ really is $\rho_0$, then the experimentalist is rewarded with (moderately) small regions\ldots but if she is wrong and $\rho$ is very different from $\rho_0$, then the assigned region will probably be so large as to imply virtually nothing about $\rho$.

So while the extreme choice $\mu = \delta(\rho-\rho_0)$ is a valid one, it's unwise in practice.  It \emph{does} play an important role by establishing an absolute (and tight) lower bound on $\Vbar(\rho_0)$.  But in practice, any sane experimentalist or analyst will choose a more balanced estimator -- one that performs well even in the worst case.  This (as noted above) is the minimax estimator, and is the PR estimator for the least favorable prior.  

Finding exact LFPs is arduous and tricky at best\cite{Evans}.  Worse yet, the LFP will depend (perhaps sensitively) on the exact volume measure $\diff\rho$.  In order to circumvent this task (which remains a good challenge for future research), let's apply a simple heuristic ansatz to choose $\mu$.

Suppose that we choose some $\mu$, and then $\rho$ is chosen by an adversary so as to maximize $\Vbar(\rho)$.  To do so, the adversary would look for a dataset $D_0$ and a state $\rho$ such that $Pr(D_0|\rho) \gg Pr(D_0)$ (the latter is determined by $P(\rho)\diff\rho$).  If $Pr(D_0)$ is relatively small, then its ``cost'' will be relatively low, and many states $\rho'$ will be connected to it -- which means that $\Rhat(D_0)$ will be large.  But because $Pr(D_0|\rho)$ is relatively large, $D_0$ will occur relatively often, and $\Vbar(\rho)$ will be large.

Avoiding this vulnerability is simple:  ensure that $Pr(D|\rho)/Pr(D)$ is not too large for any $\rho$.  This means choosing $\mu$ so that
\begin{equation}
Pr(D) \propto \max_\rho{Pr(D|\rho)}.
\end{equation}
With this choice, the probability ratio statistic (Eq. \ref{eq:PR}) becomes the \emph{likelihood ratio} statistic,
\begin{equation}
r = \frac{Pr(D|\rho)}{Pr(D)} \to \frac{Pr(D|\rho)}{\max_{\rho'}{Pr(D|\rho')}} = \frac{\cL(\rho)}{\cL_{\mathrm{max}}(\rho)} \equiv \Lambda,
\end{equation}
which defines likelihood-ratio (LR) regions (Definition \ref{def:LR}) as a special case of PR regions.  We use $\lambda = -2\log\Lambda$, rather than $\Lambda$, to maintain a convenient connection with [extensive] previous work on likelihood ratios.

\section{The threshold}

For a generic PR region estimator, the threshold value of the statistic ($r_\alpha$) depends significantly on $\rho$.  We could define LR regions in the same way, using a $\rho$-dependent threshold $\lambda_\alpha(\rho)$.  But one of the special attributes of the LR statistic is that, unlike generic PR statistics, its distribution is approximately independent of $\rho$ (as shown in Fig. \ref{fig:Cutoff}).  Proving this independence exactly depends on a Gaussian approximation that does not quite hold even as $N\to\infty$, but it does hold approximately.  So while \emph{exactly} optimal regions require calculating a $\rho$-dependent threshold numerically (see Fig. \ref{fig:Cutoff}), we can replace $\lambda_\alpha(\rho)$ with a constant lower bound $\lambda_\alpha$ satisfying
\begin{equation}
\lambda_\alpha \geq \lambda_\alpha(\rho)\mathrm{\ for\ all\ }\rho \label{eq:threshold1}
\end{equation}
and obtain a simpler and more elegant estimator that is still rigorously correct, and only sacrifices a little power.  This simplifies matters -- but we still need to set $\lambda_\alpha$!

\begin{figure}[t]
\begin{centering}
\includegraphics[width=8cm]{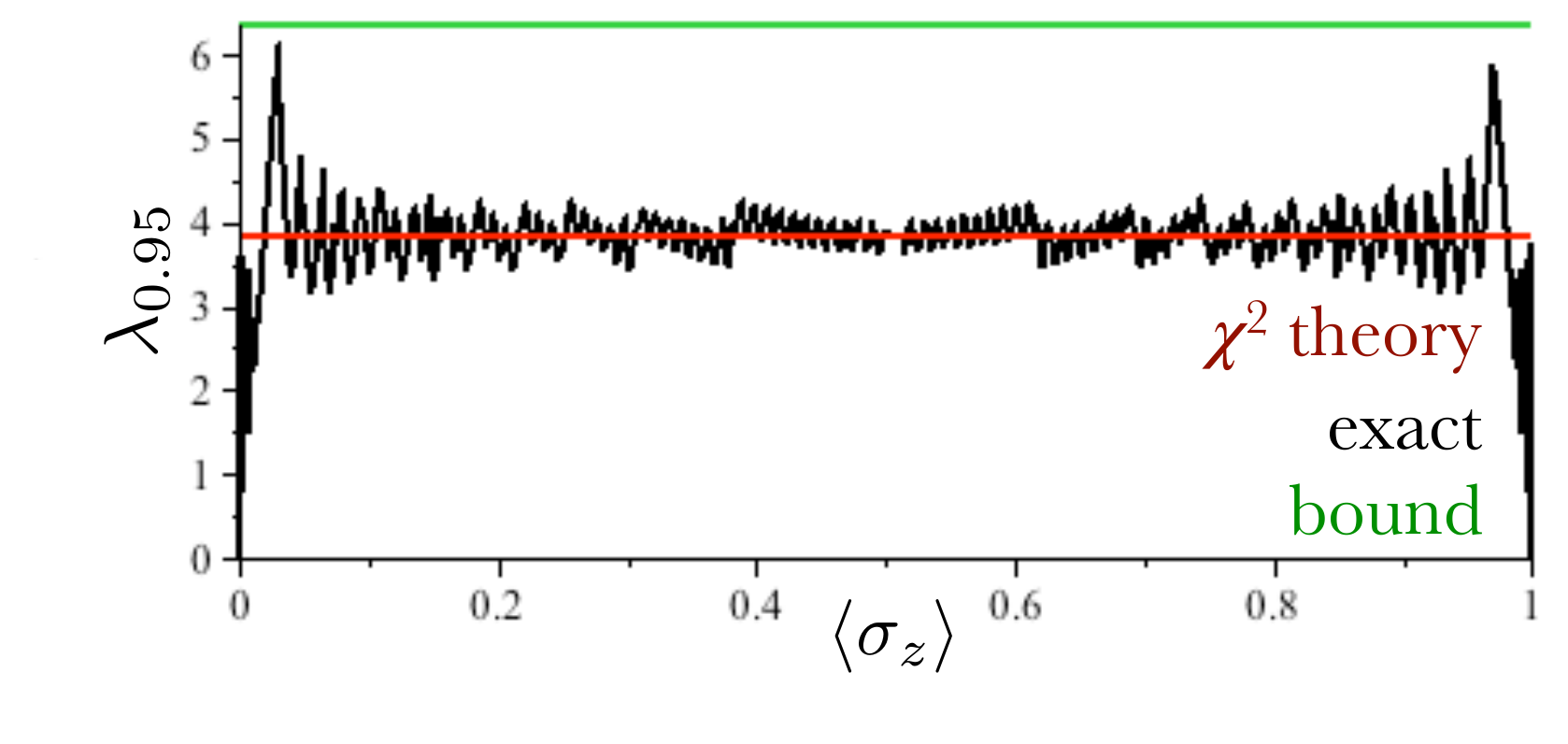}
\par\end{centering}
\caption{This figure shows the tight state-\emph{dependent} cutoff $\lambda_\alpha(\rho)$ for a particularly simple case:  a qubit measured only in the $\sigma_z$ basis (or, alternatively, a classical coin).  The cutoff depends only on $\expect{\sigma_z}$, so it can be plotted easily and compared to (i) the $\chi^2$ value, around which it fluctuates, and (ii) the upper bound given by Eq. \ref{eq:Fbound}.  While using the state-dependent cutoff would yield (slightly) smaller regions, there is a concomitant loss of simplicity, elegance, and convenience (because the regions become nonconvex).}
\label{fig:Cutoff}
\end{figure}

Coverage probability and region size both increase with $\lambda_\alpha$.  So we want to set it as low as possible (to ensure powerful regions) while maintaining coverage probability at least $\alpha$.  $\Rhat_\alpha$ will include $\rho$ if and only if $\lambda(\rho) < \lambda_\alpha$.  In principle, we can compute the distribution of $\lambda(\rho)$ for every possible $\rho$, define a \emph{complementary cumulative distribution function} (CCDF)
$$F(\lambda_\alpha|\rho) = Pr( \lambda(\rho) > \lambda_\alpha | \rho),$$
and then solve the equation $\max_\rho{ F(\lambda_\alpha|\rho) } = 1-\alpha$ for $\lambda_\alpha$.  In practice, computing $F$ is hard, so instead we use upper bounds on $F$ (as shown in Fig. \ref{fig:Threshold}) to set $\lambda_\alpha$, ensuring coverage probability $\geq\alpha$ at a small cost in power.  Two valid and useful (though not tight) bounds are given in Eq. \ref{eq:Fbound} (whose derivation is rather arduous, and will be published elsewhere) and Lemma \ref{lem:bound} (proven herein).

If the data had a Gaussian distribution, then $\lambda(\rho)$ would be a $\chi^2_k$ random variable, where $k$ (the number of degrees of freedom) is the number of linearly independent observables measured, and is equal to $d^2-1$ for an informationally complete set of measurements.  The corresponding CCDF is independent of the true state $\rho$, and given in terms of an upper incomplete Gamma function by
\begin{eqnarray}
F_{\chi^2_k}(\lambda_\alpha) &=& \frac{\gamma(k/2,\lambda_\alpha/2)}{\Gamma(k/2)} \label{eq:chi2}\\
&\approx&
\frac{1}{\left(\frac{k}{2}-1\right)!}\left(\frac{\lambda_\alpha}{2}\right)^{k/2-1}e^{-\lambda_\alpha/2},
\end{eqnarray}
where the second line is valid as $\lambda_\alpha\to\infty$.  Unfortunately, tomographic data are multinomial, not Gaussian, and this ansatz is too optimistic.  Using it yields a coverage probability that is $\alpha$ only on average, and can be much lower for some $\rho$.  A much more arduous calculation (to be published separately) yields an upper bound
\begin{eqnarray}
F(\lambda_\alpha) &\leq& F_{\chi^2_k} + e^{-\lambda_\alpha/2}\left[\left(1+\frac{\sqrt{3e\lambda_\alpha}}{\pi}\right)^{k} - \frac{\sqrt{3e\lambda_\alpha}}{\pi}^k\right] \nonumber\\
&\to& F_{\chi^2_k} + e^{-\lambda_\alpha/2}k\sqrt{\lambda_\alpha}^{k-1}\mathrm{\ when\ }\lambda_\alpha\to\infty, \label{eq:Fbound}
\end{eqnarray}
that is valid whenever the data $D$ are obtained from independent measurements on identically prepared copies of $\rho$ (the standard tomographic setup).  Figure \ref{fig:Threshold} compares these bounds with an exhaustive numerical calculation of $F(\lambda_\alpha)$ for the $k=3$ degrees of freedom found in single-qubit tomography.

\begin{figure}[t]
\begin{centering}
\includegraphics[width=8cm]{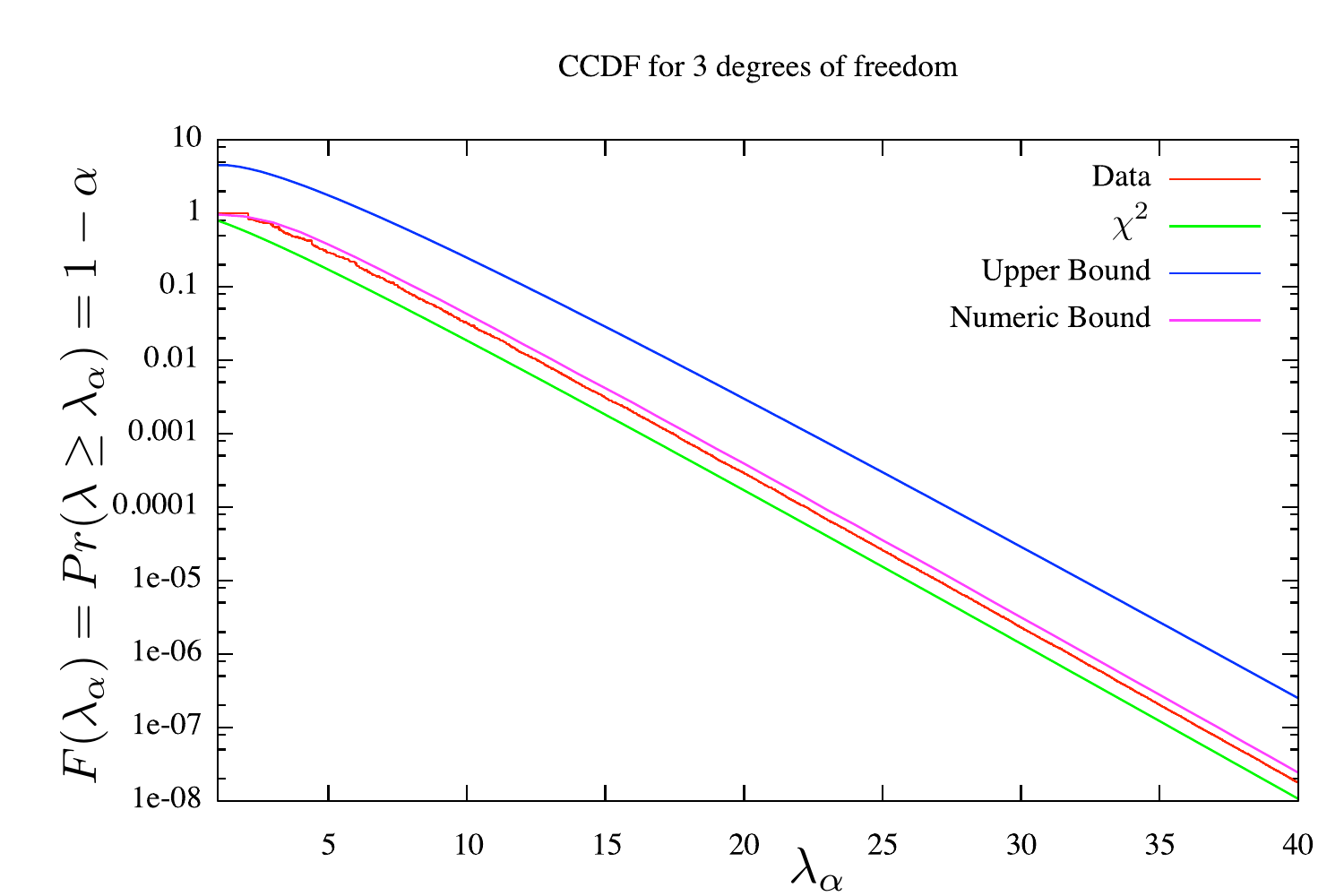}
\par\end{centering}
\caption{Three different bounds on the complementary cumulative distribution function or CCDF [$F(\lambda_\alpha)$] for the loglikelihood ratio derived from qubit state tomography (wherein $\rho$ has $K=3$ degrees of freedom).  The horizontal axis is the cutoff value $\lambda_\alpha$, while the vertical (log-scale) axis is the associated failure probability $1-\alpha = F(\lambda_\alpha)$.  ``\textbf{Data}'' correspond to an exhaustive numerical calculation of the CCDF, and are compared with the $\chi^2$ approximation (Eq. \ref{eq:chi2}) and the upper bound given in Eq. \ref{eq:Fbound}.  The plot labeled ``\textbf{Numeric bound}'' is a hybrid calculation where method used to derive Eq. \ref{eq:Fbound} is augmented by calculating one hard-to-approximate quantity numerically; its excellent agreement with data supports Eq. \ref{eq:Fbound}, and suggests that it can be improved.}
\label{fig:Threshold}
\end{figure}

A simpler (but looser) bound that applies to \emph{any} data, including joint measurements\footnote{This bound was inspired by Refs. \cite{Christandl,ChristandlPRL09}, and by discussions with Matthias Christandl, for whose help I am grateful!}, is
\begin{lemma} For any measurement on $N$ copies of a $d$-dimensional system, $F(\lambda_\alpha) \leq N^{d^2-1}e^{-\lambda_\alpha/2}$.
\label{lem:bound}
\end{lemma}
\begin{proof} 
A POVM measurement $\cM = \{E_k\}$ (where $E_k\geq0$ and $\sum_k{E_k}=\Id$) is performed on the state $\rho^{\otimes N}$, on the $N$-copy Hilbert space $\cH^{\otimes N}$.  By Schur's Lemma, we may decompose the Hilbert space as $\cH^{\otimes N} = \bigoplus_{\nu}{\cU_\nu\otimes\cP_\nu}$, where $\cU_\nu$ and $\cP_\nu$ are irreducible representation spaces of the unitary group $SU(d)$ and the permutation group $S_N$ (respectively).  Because the state is permutation-symmetric, it is maximally mixed on the $\cP_\nu$ factors, and therefore this measurement is equivalent to a measurement $\cM' = \{E'_k\}$ on the much smaller Hilbert space $\bigoplus_\nu{\cU_\nu}$, whose dimension is at most $M = N^{d^2-1}$.

Next, let the probability of event $k$ given $\rho$ be $p_k$, and the \emph{maximum} probability of event $k$ be $q_k$, so that $\lambda = -2\log( p_k / q_k )$.  Rewriting this gives $p_k = q_ke^{\lambda/2}$.  Now, $q_k\leq Tr(E'_k)$, and $\sum_k{E'_k} = \Id$, so $\sum_k{q_k} \leq M$.  The probability that $\lambda \geq \lambda_\alpha$ is $Pr( \lambda \geq \lambda_\alpha ) = \sum_{\mathrm{all\ }k\mathrm{\ s.t.\ }\lambda \geq \lambda_\alpha}{p_k}$.  Each term is $p_k = q_k e^{-\lambda/2} \leq q_k e^{-\lambda_\alpha/2}$, so the sum is upper bounded by $Me^{-\lambda_\alpha/2} = N^{d^2-1}e^{-\lambda_\alpha/2}$.  
\end{proof}
Equation \ref{eq:Fbound} does not grow with the number of samples measured ($N$), so it will generally be much tighter than Lemma \ref{lem:bound} -- but is harder to derive.  And in many cases Lemma \ref{lem:bound} may be fine; because $F(\lambda_\alpha)$ decreases exponentially in all cases, the $N^{d^2-1}$ factor will enlarge the region's size by at most $\mathrm{polylog(N)}$.

Since these bounds are loose, their use will produce confidence regions that are somewhat larger than necessary.  We will want to know \emph{how} excessively large they are!  Fortunately, such a tool is ready to hand.  The $\chi^2$ approximation to $F(\lambda_\alpha)$, given in Eq. \ref{eq:chi2}, gives a \emph{lower} bound\footnote{It's worth being precise here.  Eq. \ref{eq:chi2} is not a lower bound on $F()$.  Instead, the true exact $F()$ depends on $\rho$, and fluctuates above and below the $\chi^2$ approximation, as shown in Fig. \ref{fig:Cutoff}.  Since $F()$ is greater than the $\chi^2$ approximation for some $\rho$, any state-independent upper bound on $F()$ will be strictly greater than the $\chi^2$ approximation, so Eq. \ref{eq:chi2} is a lower bound on such upper bounds.} on $\lambda_\alpha$.  This suggests a simple test:
\begin{enumerate}
\item Define a confidence region $\Rhat(D)$ using a value of $\lambda_\alpha$ obtained by solving Eq. \ref{eq:Fbound} or the equation in Lemma \ref{lem:bound}.
\item Define an ``inner bound'' region $\Rhat_{\mathrm{min}}(D)$ using $\lambda_\alpha$ obtained by solving Eq. \ref{eq:chi2}.
\item If $\Rhat(D)$ and $\Rhat_{\mathrm{min}}(D)$ are relatively similar (i.e., $\Rhat(D)$ is not much bigger, and they lead to similar conclusions about the experiment) then there's nothing to be gained from using a tighter bound.
\end{enumerate}

\section{Discussion}

Likelihood-ratio confidence regions define ``error bars'' for quantum tomography that are:
\begin{enumerate}
\item Rigorously guaranteed to capture the true state (or process) with controllable probability $\alpha$,
\item Approximately as small as can be achieved (within that constraint),
\item Natural, convenient, and intuitive.
\end{enumerate}
The third point summarizes several particularly nice properties of LR regions.  They are convex for standard tomographic data (because the log likelihood is convex), so they can be manipulated and characterized using convex programming.  Determining whether a given state $\rho$ lies in $\Rhat(D)$ is even easier -- just compare $\cL(\rho)$ to $\cL_{\mathrm{max}}$.  And LR regions are, in a sense, a natural generalization of the popular maximum-likelihood (ML) point estimator.  This view is actually backward;  whereas ML point estimators have no finite-sample optimality properties, and may yield pathological results for quantum tomography, LR regions form a provably near-optimal region estimator.  So a more accurate view is that the near-optimality of LR regions explains why ML point estimators often work well:  the true state is usually in $\Rhat_{\alpha}(D)$, and $\Rhat_{\alpha}(D)$ is usually a neighborhood of $\rhoMLE$.

Unlike the ellipsoidal regions implied by error bars, or the spherical ones implied by large deviation bounds\cite{SugiyamaPRA11,Liu11}, LR regions have shapes that are variable and data-adapted (see Fig. \ref{fig:Examples}).  This is a virtue -- they can be much smaller than the best region of a fixed shape.  But it can also be inconvenient.  Many questions about $\rho$ can be answered directly using the simple implicit description of $\Rhat(D)$ and convex programming (e.g. ``Is $\rho$ definitely separable?'' or ``What values of $\expect{X}$ can be ruled out?'').  Sometimes, though, an explicit description of $\Rhat(D)$ is required.  One can be produced with reasonable efficiency by sampling from the surface and calculating a minimum-volume bounding ellipsoid \cite{KumarJOTA05} or hypersphere \cite{FischerLNCS03}.  This trades power (the approximated region is larger) for convenience.

Most of the error bars or regions used to date in quantum tomography are based on standard errors -- i.e., the variance of a point estimator, usually the maximum likelihood estimator.  As discussed in the introduction, such regions are not reliable (they may work in many cases, but not all!).  However, a rigorously reliable solution was proposed\cite{Christandl} quite recently by Christandl and Renner.  This excellent result deserves some discussion here.  Although it was obtained entirely independently of the current work\footnote{The authors of Ref. \cite{Christandl} and I became aware of each others' work only when our simultaneous submissions to the QIP 2012 conference were accepted and merged by the program committee.  The present paper is \emph{almost} independent of our subsequent collaboration; the exception is Lemma \ref{lem:bound}, which was proved by myself but directly inspired by conversations with Matthias Christandl at ETH-Zurich.}, it addresses a very similar problem and arrives at a solution that (in some ways) is closely related\ldots via remarkably different methods!

Their main result is that confidence regions with confidence $\alpha$ can be constructed by (i) constructing a Bayesian \emph{credible} region (a region containing $\alpha'\approx1$ of the posterior probability) for a particular prior (Hilbert-Schmidt measure over quantum states), then (ii) enlarging the region slightly in a particular way.  This proof elegantly synthesizes Bayesian and frequentist notions, by constructing regions that are simultaneously confidence regions \emph{and} credible regions (at least for a particular prior).  It suggests that the two approaches -- while philosophically orthogonal -- are deeply related in some fashion.

On the other hand, the method in Ref. \cite{Christandl} is not explicit.  That is, it does not suggest \emph{which} high-posterior-probability region to report.  The natural choice is to choose the smallest such region, but it is not immediately obvious how to identify it.  Moreoever, there is no obvious way to determine how powerful this procedure is -- i.e., whether it assigns regions that are, in some sense, smaller than those assigned by most or all other estimators.

These are exactly the strengths of the likelihood-ratio method given here.  Definition \ref{def:LR} defines a specific, simple, and straightforward protocol, and we can easily analyze the expected volume of LR regions and show that they are optimal in a strong sense.  The resulting regions are also relatively easy to characterize using known properties of the likelihood function (e.g., for standard tomographic data it is convex, and so are the LR regions).  On balance, the LR method given here seems more practically useful at the present time.  However, it's worth noting that the posterior distribution used to define regions in Ref. \cite{Christandl} is \emph{very} closely related to the likelihood $\cL(\rho)$ (thanks to Bayes' Rule).  So, under most circumstances, regions assigned through a sensible interpretation of Ref. \cite{Christandl} will be quite similar to LR regions (albeit perhaps a bit larger).  This suggest that further research may synthesize both approaches into a single, uniquely satisfying definition of ``error bars''.

\vspace{0.1in}\noindent\textbf{ACKNOWLEDGEMENTS}:  This work was supported by the US Department of Energy through the LANL/LDRD program.  I am grateful for helpful conversations over a long period of time with colleagues including (but not limited to) Chris Ferrie, Josh Coombs, Steve Flammia, Carl Caves, and Aephraim Steinberg.  I am especially indebted to Matthias Christandl and Renato Renner for discussions about their closely related work, which inspired Lemma 1.


\end{document}